\documentclass[submission]{eptcs}
\providecommand{\event}{Secco 2010}
\usepackage{breakurl}
\usepackage{epsfig}
\usepackage{latexsym}
\usepackage{amsmath}
\usepackage{amsthm}
\usepackage{amssymb}
\usepackage{graphicx}
\usepackage{color}
\usepackage{url}
\usepackage[english]{babel}
\usepackage{pst-all}
\usepackage{delarray}
\usepackage{tikz}
\usetikzlibrary{arrows}
\usetikzlibrary{shapes}
\usetikzlibrary{backgrounds}
\usepackage{cite}
\newtheorem{adefinizione}{Definition}[section]
\newtheorem{afact}[adefinizione]{Fact}
\newtheorem{alemma}[adefinizione]{Lemma}
\newtheorem{teorema}[adefinizione]{Theorem}
\newtheorem{corollario}[adefinizione]{Corollary}
\newtheorem{proposizione}[adefinizione]{Proposition}
\newtheorem{claim}[adefinizione]{Claim}
\newtheorem{esempio}[adefinizione]{Example}
\newenvironment{definition}{\begin{adefinizione}}{\end{adefinizione}}

\newenvironment{lemma}{\begin{alemma}}{\end{alemma}}
\newenvironment{theorem}{\begin{teorema}}{\end{teorema}}
\newenvironment{corollary}{\begin{corollario}}{\end{corollario}}
\newenvironment{proposition}{\begin{proposizione}}{\end{proposizione}}
\newenvironment{example}{\begin{esempio}}{\end{esempio}}

\newcommand{\BNDC}{\mbox{\it BNDC}}
\newcommand{\NDC}{\mbox{\it NDC}}
\newcommand{\INI}{\mbox{\it INI}}
\newcommand{\mathbox}[1]{\mbox{{\small \mbox{$ #1 $}}}}
\newcommand{\sta}[3]{\mathbox{#1 \stackrel{#2}{\longrightarrow} #3}}
\def\bbbn{{\rm I\!N}}
\title{On the Decidability of Non Interference over Unbounded Petri Nets}
\author{Eike Best
\institute{Universit\"at Oldenburg,
26111 Oldenburg, Germany\\
\email{eike.best@informatik.uni-oldenburg.de}}
\and
Philippe Darondeau
\institute{Inria Rennes - Bretagne Atlantique,
Rennes, France\\
\email{Philippe.Darondeau@inria.fr}}
\and
Roberto Gorrieri 
\institute{Dipartimento di Scienze dell'Informazione,
Universit\`a di Bologna, Bologna, Italy\\
\email{gorrieri@cs.unibo.it}}
}
\def\titlerunning{Deciding NI over Unbounded PN}
\def\authorrunning{E. Best, P. Darondeau \& R. Gorrieri}
\begin{document}
\maketitle

\begin{abstract}
Non-interference, in transitive or intransitive form, is defined here 
over unbounded (Place/Transition) Petri nets. The definitions are 
adaptations of similar, well-accepted definitions introduced earlier in 
the framework of labelled transition systems \cite{FG95,FG01,GV09}. The 
interpretation of intransitive non-interference which we propose for 
Petri nets is as follows. A Petri net represents the composition of a 
controlled and a controller systems, possibly sharing places and 
transitions. Low transitions represent local actions of the controlled 
system, high transitions represent local decisions of the controller, 
and downgrading transitions represent synchronized actions of both 
components. Intransitive non-interference means the impossibility for 
the controlled system to follow any local strategy that would force or 
dodge synchronized actions depending upon the decisions taken by the 
controller after the last synchronized action. The fact that both 
language equivalence and bisimulation equivalence are undecidable for 
unbounded labelled Petri nets might be seen as an indication that 
non-interference properties based on these equivalences cannot be 
decided. We prove the opposite, providing results of decidability of 
non-interference over a representative class of infinite state systems.
\end{abstract}

\section{Introduction}

Non-interference has been defined in the literature as an extensional property based on some observational semantics: the high part $H$ (i.e., the secret part) of a system does not interfere with the low part $L$ (i.e., the public part) if whatever is done in $H$ produces {\em no visible effect} on $L$. The original notion of non-interference in \cite{GM} was defined, using language equivalence, for deterministic automata with outputs. Generalized notions of non-interference were then designed to include (nondeterministic) labelled transition systems and finer notions of observational semantics such as bisimulation
(see, e.g., \cite{Ryan,FG95,RS,FG01,McC,WJ}). Recently, the problem of defining suitable non-interference properties has been attacked also in the classical model of elementary Petri nets, a special class of Petri nets where places can contain at most one token \cite{BG04apn,BG09}. When it is necessary to declassify information (e.g., when a secret plan has to be made public for realization), the two-level approach (secret/public -- $H$/$L$) is usually extended with one intermediate level of downgrading ($D$), so that the high actions that have been performed prior to a declassifying action are made public by this declassifying action. This security policy is known under the name of {\em intransitive} noninterference \cite{Rus92} (\INI\ for short) because the information flow relation is considered not transitive: even if information flows from $H$ to $D$ and from $D$ to $L$ are allowed, direct flows from $H$ to $L$ are forbidden. In \cite{GV09} intransitive non-interference has been defined for elementary net systems. 

The technical goal of this paper is to show the decidability of intransitive non-interference in the extended framework of unbounded (Place/Transition) Petri nets, and this for both definitions based alternatively on language equivalence or on weak bisimulation equivalence. As both equivalences are undecidable for unbounded labelled Petri nets \cite{Hack} \cite{Jancar}, the decidability of intransitive non-interference is not a trivial result. This is however not the first result of this type for infinite-state systems. It was actually shown in \cite{Dam} that Strong Low Bisimulation and Strong Security which is based on the latter equivalence can be decided for {\em Parallel While Programs} defined over expressions from decidable first order theories. Decidability is also established in \cite{Dam} for Strong Dynamic Security that takes both downgrading and upgrading into account. In that work, decidability comes for a large part from the property of Strong Low Bisimulation to envisage implicitly through its recursive definition all possible modifications of the dynamic store by a concurrent context (without any effective definition). In our work, decidability comes also for a large part from the fact that our basic security properties are \NDC\ ({\em NonDeducibility on Composition}) and its bisimulation version \BNDC\ \cite{FG95,FG01}, hence we envisage implicitly arbitrary concurrent contexts defined by Petri nets with high-level transitions. Now, the results presented in \cite{Dam} concern language based security whereas our results concern discrete event systems security.
As a matter of fact, both settings do not compare: on the one hand,
owing to the impossibility of testing places for zero, unbounded
Place/Transition nets have less computing power than Parallel
Write Programs, but on the other hand they have {\em labeled}
transition semantics whereas Parallel Write Programs have
{\em unlabeled} transition semantics.

Let us now explain the meaning of non-interference in the context of 
systems and control. In the Ramadge and Wonham approach to supervisory 
control for safety properties of discrete event systems 
\cite{RW87,RW89}, one considers closed loop systems made of a plant (the 
system under control) and a controller that may share actions but have 
disjoint sets of local states. Synchronization on shared actions allows 
the controller to observe the plant and to disable selected actions of 
the plant. Actions of the plant may be invisible to the controller, but 
all actions of the controller are shared with the plant and 
synchronized. Moreover controllers are deterministic, hence the current 
state of the controller may be inferred from the past behaviour of the 
plant. In the present paper, the closed system made of the plant and the 
controller is modelled by an unbounded Petri net with three levels of 
transitions $L$, $D$ and $H$. A place may count e.g. an unbounded number 
of clients or goods. Transitions in $L$ represent actions of the plant 
alone. Transitions in $D$ represent synchronized actions of the plant 
and the controller. Transitions in $H$ represent actions of the 
controller alone. Here the controller can check and modify proactively 
the global state to orient runs towards reaching some set of states or 
to maximize some profit. Intransitive non-interference means the 
impossibility for the controlled system, seen as the adversary of the 
controller, to win by forcing or dodging synchronized actions that 
depend upon the decisions taken by the controller after the last 
synchronized action. An example is given in Section~\ref{intransitive}.
    
We are mainly interested in intransitive non-interference. Nevertheless, in a large part of the paper, we shall focus on classical non-interference, in order to establish first the technical results in a simpler framework. In Section \ref{back} we recall the basics of labeled transition systems and Petri nets. Section \ref{noninterference} presents the definitions of classical non-interference notions for PT-nets, and proves that both language equivalence and weak bisimulation equivalence based notions of classical non-interference are decidable. Section \ref{intransitive} presents the definition of intransitive non-interference for PT-nets, introduces examples showing the practical significance of this notion in the context of discrete event systems, and provides decidability results extending the results of Section \ref{noninterference}. Section \ref{conc} reports some conclusive remarks. A short appendix recalls some results on Petri nets and semi-linear sets used in our proofs.

\section{Background}\label{back}

\subsection{Transition systems and bisimulations}

\begin{definition}[LTS]
A {\em labeled transition system} over a set
of {\em labels} $\Sigma$ is a tuple 
$\mathcal{T}=(Q,T,q_0)$ where $Q$ is a set of 
{\em states}, $q_0\in Q$ is the initial state, 
and $T\subseteq Q\times\Sigma\times Q$ is a set 
of {\em labeled transitions}. An LTS is said 
to be {\em deterministic} if $(q,\sigma,q')\in T$ 
and $(q,\sigma,q'')\in T$ entail $q'=q''$. 
\end{definition}

\begin{definition}[LTS under partial observation]
A partially observed LTS is an LTS 
$\mathcal{T}=(Q,T,q_0)$ over a set of labels 
$\Sigma$ which is partitioned into 
{\em observable} labels $\sigma\in\Sigma_o$ 
(for convenience, we assume that
$\varepsilon\notin\Sigma_o$) and 
{\em unobservable} labels $\tau\in\Sigma_{uo}$.
In a partially observed LTS,
$q\rightarrow^*q'$ denotes the least binary 
relation on states such that $q\rightarrow^*q$ for 
all $q\in Q$, $q\rightarrow^*q'$ for all 
$(q,\tau,q')\in T$ with $\tau\in\Sigma_{uo}$, and 
$q\rightarrow^*q'$ whenever $q\rightarrow^*q''$ 
and $q''\rightarrow^*q'$ for some $q''$. 
\end{definition}

\begin{definition}[Language equivalence]\label{LE1}
The {\em language} of a partially observed
LTS is the set of all finite words 
$\sigma_1\sigma_2\ldots\sigma_n$ (including 
$\varepsilon$ which corresponds to $n=0$) such 
that 
$q_0\rightarrow^*\sta{q_1}{\sigma_1}{q'_1}
\rightarrow^*\sta{q_2}{\sigma_2}{q'_2}\ldots
\rightarrow^*\sta{q_n}{\sigma_n}{q'_n}$ for some
adequate sequence of states $q_i$ and $q'_i$.
Two partially observed LTS's $\mathcal{T}$ and
$\mathcal{T}'$ are {\em language equivalent}
(in notation, $\mathcal{T}\sim\mathcal{T'}$) 
if they have the same language.  
\end{definition}

\begin{definition}[Weak simulation]\label{wsim.def}
Given a set of labels $\Sigma=\Sigma_o\cup\Sigma_{uo}$
and two partially observed LTS's $\mathcal{T}$ and 
$\mathcal{T}'$ over $\Sigma$, $\mathcal{T}$ is
{\em weakly simulated} by $\mathcal{T}'$ (or $\mathcal{T}'$
weakly simulates $\mathcal{T}$) if there exists a binary 
relation $R\subseteq Q\times Q'$, called a {\em weak simulation},
such that $(q_0,q'_0)\in R$ and the 
following requirements are satisfied 
for all $(q_1,q'_1)\in R$, and for all 
$\sigma\in\Sigma_o$ and $\tau\in\Sigma_{uo}$:
\begin{itemize}
\item
if $\sta{q_1}{\sigma}{q_2}$ then $(\exists q'_2):\,
(q_2,q'_2)\in R$ and $q'_1\rightarrow^*
\sta{q''_1}{\sigma}{}
q''_2\rightarrow^* q'_2$, 
\item
if $(q_1,\tau,q_2)\in T$ then $(\exists q'_2):\,
(q_2,q'_2)\in R$ and
$q'_1\rightarrow^* q'_2$.
\end{itemize}
\end{definition}

If $\mathcal{T}$ is simulated by $\mathcal{T}'$, 
then the language of $\mathcal{T}$ is included in 
the language of $\mathcal{T}'$. 

\begin{definition}[Weak bisimilarity]\label{WB1}
Given a set of labels $\Sigma=\Sigma_o\cup\Sigma_{uo}$,
two partially observed LTS's $\mathcal{T}=(Q,T,q_0)$ and 
$\mathcal{T}'=(Q',T',q'_0)$ over $\Sigma$ are {\em weakly 
bisimilar} (in notation, $\mathcal{T}\approx\mathcal{T'}$) 
if and only if there exists some binary relation 
$R\subseteq Q\times Q'$, called a {\em weak bisimulation},
such that $(q_0,q'_0)\in R$ and both $R$ and $R^{-1}$ are
weak simulations.
\end{definition}

If $\mathcal{T}$ 
and $\mathcal{T}'$ are weakly bisimilar, then they
are language equivalent. 


\subsection{Place/Transition Petri nets}

In order to keep the presentation concise, we omit here 
the basic definition of Petri nets which may be found 
in an appendix together with some classical decidability 
results.

\begin{definition}[PT-net system]
A {\em PT-net system} $\mathcal{N}=(P,T,F,M_0)$ is a 
PT-net with an {\em initial marking} $M_0$. The {\em
reachability set} $RS(\mathcal{N})$ of $\mathcal{N}$ 
is the set of all markings that may be reached from
$M_0$ by sequences of transitions of the net. The 
{\em reachability graph} $RG(\mathcal{N})$ of $\mathcal{N}$ 
is the LTS with the set of states $[M_0\rangle$ and
the initial state $M_0$, where $[M_0\rangle=
RS(\mathcal{N})$ and there is a 
transition from $M$ to $M'$ labeled with $t$ if{}f 
$M[t\rangle M'$. Given $\mathcal{N}=(P,T,F,M_0)$, 
the {\em underlying net} is $\mathcal{U}(\mathcal{N})
=(P,T,F)$. For convenience, we write
$\mathcal{N}=(\mathcal{U}(\mathcal{N}),M_0)$. 
\end{definition}

\begin{definition}[Composition of net systems]\label{netcom.def}
Given two PT-net systems
$\mathcal{N}_1=(P_1,T_1,F_1,M_{1,0})$ and 
$\mathcal{N}_2=(P_2,T_2,F_2,M_{2,0})$ such that
$P_1 \cap P_2 = \emptyset$, their 
composition $\mathcal{N}_1\,|\,\mathcal{N}_2$
is the  PT-net system 
$(P,T,F,M_{0})$ where $P$ is the 
union of $P_1$ and $P_2$, $T$ is the union of $T_1$ 
and $T_2$, and $F$ and $M_0$ are the unions of the
maps $F_i$ and $M_{i,0}$ respectively, for $i=1,2$.
Also let 
$\mathcal{U}(\mathcal{N}_1)\,|\,
\mathcal{U}(\mathcal{N}_2)=$
$\mathcal{U}(\mathcal{N}_1\,|\,\mathcal{N}_2)$. 
\end{definition}

Note that synchronisation occurs over those transitions that are 
shared by the two nets, that is, for a transition $t$ that occurs both in $T_1$ and $T_2$,
we have that, e.g.,  $F(p,t) = F_1(p, t)$ if $p \in P_1$, $F(p,t) = F_2(p, t)$ otherwise.

\begin{definition}[Restriction of a net system]
Given a PT-net system $\mathcal{N}=(P,T,F,M_{0})$
and a subset of transitions $T'\subseteq T$, let
$\mathcal{N}\setminus T'=(P,T\setminus T',F',M_{0})$ 
where $F'$ is the induced restriction of $F$ on 
$T\setminus T'$. Also let   
$\mathcal{U}(\mathcal{N})\setminus T'=
(P,T\setminus T',F')$.
\end{definition}

\begin{definition}[Labeled net system]
A {\em labeled net system} $(\mathcal{N},\lambda)$
is a PT-net system $\mathcal{N}=(P,T,F,M_{0})$ with
a {\em transition labelling map} 
$\lambda:T\rightarrow\Sigma_o\cup\{\varepsilon\}$
(the subscript $o$ in $\Sigma_o$ means an alphabet 
of observations).
 The {\em labeled reachability graph} of
$(\mathcal{N},\lambda)$ is the partially observed
LTS over $\Sigma=\Sigma_o\cup\{\varepsilon\}$ 
which derives from $RG(\mathcal{N})$ by replacing each
transition $M[t\rangle M'$ with a corresponding 
transition $(M,\lambda(t),M')$. 
\end{definition}

\begin{definition}[Weak simulation]
Given two labeled net systems $(\mathcal{N},\lambda)$ 
and $(\mathcal{N}',\lambda')$ over the same set of 
labels $\Sigma_o$, $(\mathcal{N},\lambda)$ is weakly
simulated by $(\mathcal{N}',\lambda')$ if the labeled
reachability graph of $\mathcal{N}$ is weakly simulated 
by the labeled reachability graph of $\mathcal{N}'$.
\end{definition}

\begin{definition}[Equivalences of labeled net systems]
Two labeled net systems $(\mathcal{N},\lambda)$ 
and $(\mathcal{N}',\lambda')$ over the same set 
of labels $\Sigma_o$ are:
\begin{itemize} 
\item
{\em language equivalent}
(in notation, $(\mathcal{N},\lambda)\sim
(\mathcal{N}',\lambda')$ or for short $\mathcal{N}\sim
\mathcal{N}'$ when the labelling maps are clear from the
context)
if their labeled reachability graphs are
language equivalent;
\item
{\em weakly bisimilar}
(in notation, $(\mathcal{N},\lambda)\approx
(\mathcal{N}',\lambda')$ or for short $\mathcal{N}\approx
\mathcal{N}'$ when the labelling maps are clear from the
context)
if their labeled reachability graphs are
weakly bisimilar.
\end{itemize}
A weak bisimulation between the labeled 
reachability graphs of two labeled net systems
is called a weak bisimulation between them.
\end{definition}

A particular case is with {\em partially observed net systems}, i.e. 
when $\Sigma_o=T_o\subseteq T$, $\lambda(t)=t$ for $t\in T_o$,
and $\lambda(t)=\varepsilon$ for $t\in T\setminus T_o$. 
For partially observed net systems, 
$(\mathcal{N},\lambda)\sim(\mathcal{N}',\lambda')$
if and only if the reachability graphs of $\mathcal{N}$ and $\mathcal{N}'$, 
considered as partially observed LTS's with $\Sigma_{uo}=T\setminus T_o$,
are language equivalent  in
the sense of Definition~\ref{LE1}.
In the same conditions,
$(\mathcal{N},\lambda)\approx (\mathcal{N}',\lambda')$ if and only if 
$RG(\mathcal{N})\approx RG(\mathcal{N}')$ in the sense of Definition~\ref{WB1}.

\begin{proposition}
\label{sim=approx}
If $\lambda$ is the identity, $(\mathcal{N},\lambda) \approx (\mathcal{N}',\lambda)$
iff $(\mathcal{N},\lambda) \sim (\mathcal{N}',\lambda)$
\end{proposition}

\section{Classical non-interference in PT-nets}\label{noninterference}

In this section, we focus on systems that can perform two kinds of actions: high-level actions, representing the interaction of the system with high-level users, and low-level actions, representing the interaction of the system with low-level users. The system has the property of non-interference if the interplay between its low-level part and high-level part cannot affect the low level user's view of the system, even assuming that the low-level user knows the structure of the system. As already said in the introduction, the goal of this section is to provide the technical basis that we need for showing subsequently the decidability of intransitive non-interference for PT-nets, which we feel has more direct interest for applications in the context of discrete event systems. We must therefore postpone the presentation of motivating examples.    

\begin{definition}[Two-level net system]
A two-level PT-net system is a PT-net system 
$\mathcal{N}=(P,T,F,M_0)$ whose set of transitions  
$T$ is partitioned into {\em low level} 
transitions $l\in L$ and {\em high level} 
transitions $h\in H$, such that $T=L\cup H$ and 
$L\cap H=\emptyset$. A net system $\mathcal{N}$ 
is a {\em high-level net system} if all transitions 
in $T$ are high-level transitions. It is a 
{\em low-level net system} if all transitions in $T$ 
are low-level transitions.
\end{definition} 

Henceforth, {\em two-level net systems are considered 
as partially observed net systems} where the transitions 
in $L$ are observable while the transitions in $H$ are 
unobservable ($\Sigma_o=L$ and $\Sigma_{uo}=H$). This 
interpretation applies to all instances of the relations  
$\mathcal{N}\sim \mathcal{N}'$ or 
$\mathcal{N}\approx \mathcal{N}'$ between two-level 
net systems. We denote by $\mathcal{L}(\mathcal{N})$ the
language of a two-level net system $\mathcal{N}$, that 
is to say, the set of images $\lambda(t_1t_2\ldots t_n)$ 
of sequences of transitions $M_0[t_1t_2\ldots t_n\rangle
M$ under the labelling map $\lambda(t)=t$ for $t\in L$ 
and $\lambda(t)=\varepsilon$ for $t\in H$.  

\begin{definition}[NDC-BNDC]\label{bndc} 
A {\em two-level} net system $\mathcal{N}$ has the property 
NDC (Non-Deducibility on Compositions),
resp. BNDC 
(Bisimulation-Based Non-Deducibility on Compositions), 
if for any high-level net system 
$\mathcal{N'}$ with a set of transitions $H'$ not 
intersecting $L$, the two-level net systems 
$\mathcal{N}\setminus H$ and 
$(\mathcal{N}|\,\mathcal{N'})\setminus(H\setminus H')$
are language equivalent, resp. weakly bisimilar. 
\end{definition}

The definitions of NDC and BNDC are very strong, and 
their verification is indeed quite demanding: infinitely many 
equivalence checks are required, one for each choice
of a high-level net system $\mathcal{N'}$. Moreover, each
equivalence check may be a problem, as both language 
equivalence and bisimulation equivalence are undecidable 
over unbounded labeled PT-nets and likewise over 
unbounded partially observed 
PT-nets \cite{Hack,Jancar}.
We shall discuss about the strength of
these notions in section~\ref{intransitive}. For the 
moment, what we need is an alternative characterization 
of these properties, more amenable for an algorithmic 
treatment in view of showing decidability.

\subsection{Deciding on NDC}

In this section, we show that $\mathcal{N}$ enjoys NDC 
if and only if $\mathcal{N}$ and $\mathcal{N}\setminus H$
are language equivalent.

\begin{proposition}\label{simndc}
For any high-level net system $\mathcal{N'}$ with set
of transitions $H'$ not intersecting $L$, $\mathcal{N}\setminus H$ is 
weakly simulated by 
$(\mathcal{N}|\,\mathcal{N'})\setminus(H\setminus H')$
which in turn is weakly simulated by $\mathcal{N}$ 
(where all net systems under consideration have the same set of
observable transitions $\Sigma_o=L$).
\end{proposition}

\begin{proof}
Any transition from $L$ has similar place neighbourhoods in
$\mathcal{N}\setminus H$, 
$(\mathcal{N}|\,\mathcal{N'})\setminus(H\setminus H')$ and
$\mathcal{N}$, and the transitions from $L$ and $H'$ have
disjoint place neighbourhoods in 
$(\mathcal{N}|\,\mathcal{N'})\setminus(H\setminus H')$. 
\end{proof}

\begin{proposition}\label{pr-ndc}
$\mathcal{N}$ has the property NDC iff $\mathcal{N}\sim
\mathcal{N}\setminus H$. Moreover, this property can be decided.
\end{proposition}

\begin{proof}
By definition, $\mathcal{N}$ has the property NDC iff,
for any high-level net system $\mathcal{N'}$ with a set 
of transitions $H'$ not intersecting $L$, the two-level 
net systems $\mathcal{N}\setminus H$ and
$(\mathcal{N}|\,\mathcal{N'})\setminus(H\setminus H')$
are language equivalent. 
Now, the chain of inclusion relations
$\mathcal{L}(\mathcal{N}\setminus H)\subseteq$
$\mathcal{L}((\mathcal{N}|\,\mathcal{N'})\setminus
(H\setminus H'))\subseteq$ $\mathcal{L}(\mathcal{N})$ holds for Proposition \ref{simndc}.
Both bounds are reached for some net system 
$\mathcal{N'}$; indeed, the lower bound is reached when 
$\mathcal{N'}$ has no place and $H'=\emptyset$, and the 
upper bound is reached when $\mathcal{N'}$ has no place 
and $H'=H$. Suppose $\mathcal{N}$ has the property NDC,
then $\mathcal{L}(\mathcal{N}\setminus H)=$
$\mathcal{L}((\mathcal{N}|\,\mathcal{N'})\setminus
(H\setminus H'))=$ $\mathcal{L}(\mathcal{N})$ for 
$\mathcal{N'}$ realizing the upper bound.
Conversely, suppose that 
$\mathcal{L}(\mathcal{N}\setminus H)=$
$\mathcal{L}(\mathcal{N})$, then necessarily
$\mathcal{L}(\mathcal{N}\setminus H)=$
$\mathcal{L}((\mathcal{N}|\,\mathcal{N'})\setminus
(H\setminus H'))$. Hence, the first claim in the proposition has been 
established. As all transitions are observable in the 
net system $\mathcal{N}\setminus H$, the language 
$\mathcal{L}(\mathcal{N}\setminus H)$ is a free Petri net 
language. By E. Pelz's theorem and corollary (Theorem~\ref{Pelz} 
in the appendix),  
one can decide whether $\mathcal{L}(\mathcal{N})\subseteq$ 
$\mathcal{L}(\mathcal{N}\setminus H)$, and hence whether 
the two languages are equal.
\end{proof}



\begin{example}\label{simple}
The net system $\mathcal{N}_1$ of Figure \ref{simpleHL}(a) 
is insecure, as $\mathcal{N}_1$ can perform the low transition $l$ at some stage, while $\mathcal{N}_1\setminus H$ cannot. On the contrary,  the net system $\mathcal{N}_2$ in 
Figure \ref{simpleHL}(b) enjoys NDC.
\begin{figure}[htbp]
\begin{center}
\begin{tikzpicture}[scale=0.7]
\node[circle,draw,minimum size=0.5cm](s1)at(0,1){};\filldraw[black](0,1)circle(3pt);
\node[draw,minimum size=0.5cm](ha)at(1.5,1){$h$};
\node[circle,draw,minimum size=0.5cm](s)at(3,1){};\draw(3,0)node{$s$};
\node[draw,minimum size=0.5cm](la)at(4.5,1){$l$};
\draw[-latex](s1.east)--(ha.west);
\draw[-latex](ha.east)--(s.west);
\draw[-latex](s.east)--(la.west);
\draw(-2,1)node{$(a):$};
\node[circle,draw,minimum size=0.5cm](s2)at(10,1){};\filldraw[black](10,1)circle(3pt);
\node[draw,minimum size=0.5cm](hb)at(11.5,1){$l$};
\node[circle,draw,minimum size=0.5cm](sd)at(13,1){};\draw(13,0)node{$s'$};
\node[draw,minimum size=0.5cm](lb)at(14.5,1){$h$};
\draw[-latex](s2.east)--(hb.west);
\draw[-latex](hb.east)--(sd.west);
\draw[-latex](sd.east)--(lb.west);
\draw(8,1)node{$(b):$};
\end{tikzpicture}
\end{center}
\caption{Two simple two-level net systems}\label{simpleHL}
\end{figure}
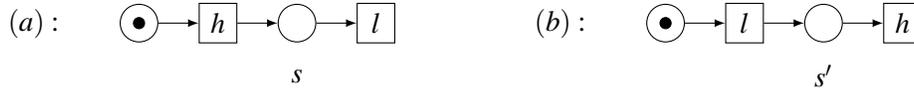
\end{example}

\begin{example}
Consider the disconnected net system $\mathcal{N}$ in Figure \ref{semicounter} (l.h.s.).
Intuitively, we expect that this system is secure
because the high part of the net (the left part) and the low part of the net (the right part) are disconnected
and so it appears that no interference is possible. In view of Definition~\ref{bndc},
it seems however difficult to verify this property by direct inspection of the infinite labeled reachability graph of $\mathcal{N}$ shown   
in Figure \ref{semicounter} (r.h.s.). With the help of Proposition \ref{pr-ndc}, this
verification becomes straightforward: the transition system that generates the
language $\mathcal{L}(\mathcal{N}\setminus H)$, which corresponds with the left column of the picture, 
and the deterministic transition system that generates the language 
$\mathcal{L}(\mathcal{N})$ (obtained by replacing all labels $h_i$ 
by $\varepsilon$ and then applying the usual subset construction)
are indeed identical. 

\begin{figure}
\begin{center}
\begin{tikzpicture}[scale=0.7]
\node[circle,draw,minimum size=0.5cm](s1)at(1,5.5){};\filldraw[black](1,5.5)circle(3pt);
\node[draw,minimum size=0.5cm](h1)at(1,4){$h_1$};
\node[circle,draw,minimum size=0.5cm](s)at(1,2.5){};
\node[draw,minimum size=0.5cm](h2)at(1,1){$h_2$};
\draw[-latex](s1.west)to[out=200,in=160](h1.west);
\draw[-latex](h1.east)to[out=20,in=-20](s1.east);
\draw[-latex](h1.south)--(s.north);
\draw[-latex](s.south)--(h2.north);
\node[circle,draw,minimum size=0.5cm](s2)at(3,5.5){};\filldraw[black](3,5.5)circle(3pt);
\node[draw,minimum size=0.5cm](l1)at(3,4){$l_1$};
\node[circle,draw,minimum size=0.5cm](sd)at(3,2.5){};
\node[draw,minimum size=0.5cm](l2)at(3,1){$l_2$};
\draw[-latex](s2.west)to[out=200,in=160](l1.west);
\draw[-latex](l1.east)to[out=20,in=-20](s2.east);
\draw[-latex](l1.south)--(sd.north);
\draw[-latex](sd.south)--(l2.north);
\end{tikzpicture}\hspace*{2.5cm}
\begin{tikzpicture}[scale=0.7]
\node[circle,fill=black!100,inner sep=0.1cm](s11)at(1,5){};
\node[circle,fill=black!100,inner sep=0.1cm](s21)at(4,5){};
\node[circle,fill=black!100,inner sep=0.1cm](s31)at(7,5)[label=right:$\cdots$]{};
\node[circle,fill=black!100,inner sep=0.1cm](s12)at(1,3){};
\node[circle,fill=black!100,inner sep=0.1cm](s22)at(4,3){};
\node[circle,fill=black!100,inner sep=0.1cm](s32)at(7,3)[label=right:$\cdots$,label=below:$\vdots$]{};
\node[circle,fill=black!100,inner sep=0.1cm](s13)at(1,1)[label=below:$\vdots$]{};
\node[circle,fill=black!100,inner sep=0.1cm](s23)at(4,1)[label=right:$\cdots$,label=below:$\vdots$]{};
\draw[-triangle 45](s11.south east)to[out=-20,in=200]node[auto]{$h_1$}(s21.south west);
\draw[-triangle 45](s21.south east)to[out=-20,in=200]node[auto]{$h_1$}(s31.south west);
\draw[-triangle 45](s21.north west)to[out=160,in=20]node[auto,swap]{$h_2$}(s11.north east);
\draw[-triangle 45](s31.north west)to[out=160,in=20]node[auto,swap]{$h_2$}(s21.north east);
\draw[-triangle 45](s12.south east)to[out=-20,in=200]node[auto]{$h_1$}(s22.south west);
\draw[-triangle 45](s22.south east)to[out=-20,in=200]node[auto]{$h_1$}(s32.south west);
\draw[-triangle 45](s22.north west)to[out=160,in=20]node[auto,swap]{$h_2$}(s12.north east);
\draw[-triangle 45](s32.north west)to[out=160,in=20]node[auto,swap]{$h_2$}(s22.north east);
\draw[-triangle 45](s13.south east)to[out=-20,in=200]node[auto]{$h_1$}(s23.south west);
\draw[-triangle 45](s23.north west)to[out=160,in=20]node[auto,swap]{$h_2$}(s13.north east);
\draw[-triangle 45](s11.south west)to[out=240,in=120]node[auto,swap]{$l_1$}(s12.north west);
\draw[-triangle 45](s12.north east)to[out=60,in=-60]node[auto]{$l_2$}(s11.south east);
\draw[-triangle 45](s21.south west)to[out=240,in=120]node[auto,swap]{$l_1$}(s22.north west);
\draw[-triangle 45](s22.north east)to[out=60,in=-60]node[auto]{$l_2$}(s21.south east);
\draw[-triangle 45](s31.south west)to[out=240,in=120]node[auto,swap]{$l_1$}(s32.north west);
\draw[-triangle 45](s32.north east)to[out=60,in=-60]node[auto]{$l_2$}(s31.south east);
\draw[-triangle 45](s12.south west)to[out=240,in=120]node[auto,swap]{$l_1$}(s13.north west);
\draw[-triangle 45](s13.north east)to[out=60,in=-60]node[auto]{$l_2$}(s12.south east);
\draw[-triangle 45](s22.south west)to[out=240,in=120]node[auto,swap]{$l_1$}(s23.north west);
\draw[-triangle 45](s23.north east)to[out=60,in=-60]node[auto]{$l_2$}(s22.south east);
\draw[-triangle 45](0.3,5.7)--(s11.north west);
\end{tikzpicture}
\end{center}
\caption{An infinite-state net system (l.h.s.) and its labeled reachability graph (r.h.s.)}\label{semicounter}
\end{figure}
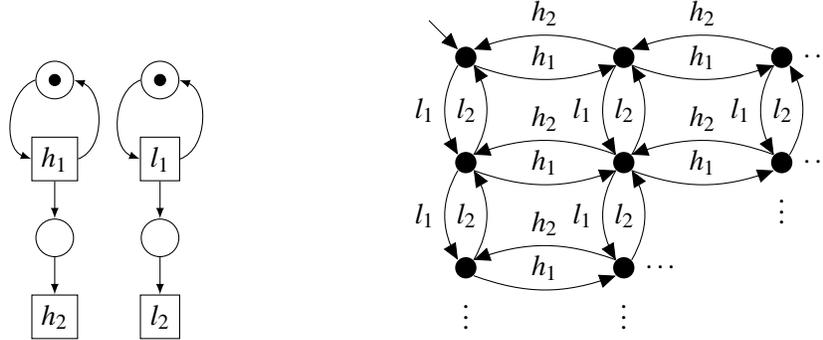
\end{example}

\subsection{Reducing BNDC to SBNDC}

For BNDC, things are a bit more complex, although we have 
the following property.

\begin{lemma}\label{lem1}
If $\mathcal{N}$ has the property BNDC, then 
$\mathcal{N}\approx\mathcal{N}\setminus H$. 
\end{lemma}

\begin{proof}
Let $\mathcal{N'}$ be the high-level net system with no 
place and with the set of transitions $H'=H$, then the
reachability graphs of $\mathcal{N}$ and
$(\mathcal{N}|\,\mathcal{N'})\setminus(H\setminus H')$
are isomorphic, hence they are weakly bisimilar, that 
is $\mathcal{N}\approx(\mathcal{N}|\,\mathcal{N'})
\setminus(H\setminus H')$. If $\mathcal{N}$
has the property BNDC, then 
$\mathcal{N}\setminus H\approx 
(\mathcal{N}|\,\mathcal{N'})\setminus(H\setminus H')$,
and the lemma follows since $\approx$ is an equivalence. 
\end{proof}

\begin{example}\label{conf}
Consider the net system $\mathcal{N}$ in Figure \ref{conflictnet}. $\mathcal{N}$ is NDC because 
$\mathcal{N} \sim \mathcal{N}\setminus H$. However, $\mathcal{N}$ is not BNDC 
because $\mathcal{N} \not\approx \mathcal{N}\setminus H$. Indeed, this net is insecure: a low-level user who is unable to perform
transition $l$ can deduce from this failure that the high-level transition $h$ has been performed.
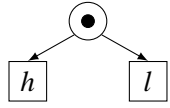
\begin{figure}
\begin{center}
\begin{tikzpicture}[scale=0.8]
\node[circle,draw,minimum size=0.5cm](s)at(2,2){};\filldraw[black](2,2)circle(3pt);
\node[draw,minimum size=0.5cm](h)at(1,1){$h$};
\node[draw,minimum size=0.5cm](l)at(3,1){$l$};
\draw[-latex](s.south west)--(h.north);
\draw[-latex](s.south east)--(l.north);
\end{tikzpicture}
\end{center}
\caption{A simple two-level net system}\label{conflictnet}
\end{figure}
\end{example}

In the rest of the section, we show that $\mathcal{N}$ enjoys 
BNDC if and only if it enjoys the property SBNDC 
defined below. 

\begin{definition}[SBNDC]\label{sbndc} 
A two-level net system $\mathcal{N}$ has the
property SBNDC (Bisimulation-Based Strong Non-Deducibility on Compositions) if, 
for any reachable marking $M_1$ of $\mathcal{N} = (N, M_0)$ and 
for any high-level transition $h\in H$, 
$M_1[h\rangle M_2$ entails that 
$(N\setminus H,M_1)$ and
$(N\setminus H,M_2)$ are 
weakly bisimilar.
\end{definition}

Note that, in view of Proposition~\ref{sim=approx}, the 
relation between $M_1$ and $M_2$ required in 
Definition~\ref{sbndc} may be equivalently expressed as
$\mathcal{L}(\mathcal{N}\setminus H,M_1)=
\mathcal{L}(\mathcal{N}\setminus H,M_2)$.

\begin{definition}\label{R}
Let $R\subseteq RS(\mathcal{N}\setminus H)\times 
RS(\mathcal{N})$ be the binary relation on 
markings which is generated from the axiom 
$M_0 R\, M_0$ by the following two inference rules
where $h\in H$ and $l\in L$:
\begin{itemize}
\item
$M_1 R\, M_2$ and $M_1=M'_1$ and 
$M_2[h\rangle M'_2$ entail 
$M'_1 R\, M'_2$
\item
$M_1 R\, M_2$ and $M_1[l\rangle M'_1$ and
$M_2[l\rangle M'_2$ entail $M'_1 R\, M'_2$
\end{itemize}
\end{definition}

Paraphrasing the definition, $M R\, M'$ 
if and only if there exist $w\in L^*$ and 
$w'\in (L\cup H)^*$ such that
$M_0[w\rangle M$, 
$M_0[w'\rangle M'$, 
and $w$ is the projection of $w'$ on $L^*$.
In the specific case where $\mathcal{N}$ 
is BNDC, $R$ is a weak bisimulation between
$\mathcal{N}\setminus H$ and $\mathcal{N}$, 
and it is indeed the {\em least} weak 
bisimulation between them.

\begin{lemma}\label{RR}
Let $\mathcal{N}=(N,M_0)$ be a net system with the 
BNDC property and let $M_1$ and $M_2$ be reachable markings of 
$\mathcal{N}\setminus H$ and $\mathcal{N}$, 
respectively. If $M_1 R\, M_2$, then
$\mathcal{L}(N\setminus H,M_1)=
\mathcal{L}(N\setminus H,M_2)$.
\end{lemma} 

\begin{proof}
As $M_1 R\, M_2$, there exist $w\in L^*$ and 
$w'\in (L\cup H)^*$ such that $M_0[w\rangle M_1$, 
$M_0[w'\rangle M_2$, and $w$ is the projection of 
$w'$ on $L^*$. Let $k=|w'|\,-|w|$ be the difference of 
length between $w'$ and $w$. Consider the high-level net 
system $\mathcal{K}=(K,M_k)$ where $K$ is a net with 
a unique place $p_k$, the set of transitions $H$, and
flow relations $F(p_k,h)=1$ and $F(h,p_k)=0$ for every 
transition $h$, and where $M_k(p_k)=k$. Let $M'_0$ and 
$M'_2$ be the markings of 
$N'=\mathcal{U}(\mathcal{N}|\,\mathcal{K})$, extending 
$M_0$ and $M_2$, respectively, such that $M'_0(p_k)=k$ 
and $M'_2(p_k)=0$. By construction, $M'_0[w'\rangle M'_2$ 
in $\mathcal{N}|\,\mathcal{K}$. As $\mathcal{N}$ has the
property BNDC and $\mathcal{K}$ is a high-level net system,
$\mathcal{N}\setminus H\approx(\mathcal{N}|\,\mathcal{K})
\setminus(H\setminus H)=$ $\mathcal{N}|\,\mathcal{K}$.
As all transitions of $\mathcal{N}\setminus H$ are 
observable and $w$ is the observable projection of $w'$,
$M_1$ and $M'_2$ are two weakly bisimilar markings of 
$\mathcal{N}\setminus H$ and $\mathcal{N}|\,\mathcal{K}$,
hence 
$\mathcal{L}(N\setminus H,M_1)=$
$\mathcal{L}(N|\,K,M'_2)$. 
As $M'_2(p_k)=0$, no transition in $H$ can occur in any
sequence fired from $M'_2$ in $N|\,K$, and therefore
$\mathcal{L}(N|\,K,M'_2)=$
$\mathcal{L}(N\setminus H,M_2)$. 
Altogether, $\mathcal{L}(N\setminus H,M_1)=$
$\mathcal{L}(N\setminus H,M_2)$.
\end{proof}

\begin{proposition}\label{elabeth}
$\mathcal{N}=(N,M_0)$ has the property BNDC {\em iff} 
for all reachable markings $M_1$ and $M_2$ of 
$N\setminus H$ and $N$, respectively, 
$M_1 R\, M_2$ entails 
$\mathcal{L}(N\setminus H,M_1)=
\mathcal{L}(N\setminus H,M_2)$.
\end{proposition}  

\begin{proof}
The direct implication has already been established. 
To show the converse implication, consider any high-level
net system $\mathcal{N'}$ with set of transitions $H'$ 
not intersecting $L$.
Let $B$ be the relation between the reachable 
markings of $\mathcal{N}\setminus H$ and
$(\mathcal{N}|\,\mathcal{N'})\setminus(H\setminus H')$
defined as follows. Let $(M_2|\,M'_2)$ denote 
the marking of $(\mathcal{N}|\,\mathcal{N'})$ that 
projects on the markings $M_2$ and $M'_2$ of 
$\mathcal{N}$ and $\mathcal{N'}$, respectively. Then, 
let $M_1 B\, (M_2|\,M'_2)$ iff $M_1 R\, M_2$.
Assume that $M_1 R\, M_2$ entails 
$\mathcal{L}(N\setminus H,M_1)=
\mathcal{L}(N\setminus H,M_2)$.
We will show that $B$ is a weak bisimulation between 
$\mathcal{N}\setminus H$ and
$(\mathcal{N}|\,\mathcal{N'})\setminus(H\setminus H')$,
entailing that $\mathcal{N}$ has the property BNDC.
As $M_1 R\, M_2$ for $M_1=M_0$ and $M_2=M_0$, the relation
$B$ holds between the initial states of the two net 
systems. 
Now consider any occurrence $M_1 B\,(M_2|\,M'_2)$ of the
relation $B$, hence $M_1 R\,M_2$ (by construction of $B$).
\begin{itemize}
\item
Let $M_1[l\rangle\widetilde{M_1}$ for $l\in L$.
As $M_1 R\,M_2$ entails 
$\mathcal{L}(N\setminus H,M_1)=
\mathcal{L}(N\setminus H,M_2)$,
necessarily, $M_2[l\rangle \widetilde{M_2}$ for some
marking $\widetilde{M_2}$, and then by definition of $R$,
$\widetilde{M_1} R\,\widetilde{M_2}$.  Thus,
$(M_2|\,M'_2)[l\rangle(\widetilde{M_2}|\,M'_2)$
with $\widetilde{M_1} B\, (\widetilde{M_2}|\,M'_2)$.
\item
Let $(M_2|\,M'_2)[l\rangle(\widetilde{M_2}|\,M'_2)$
for $l\in L$. As $M_1 R\,M_2$ entails 
$\mathcal{L}(N\setminus H,M_1)=
\mathcal{L}(N\setminus H,M_2)$, necessarily
$M_1[l\rangle\widetilde{M_1}$ for some marking 
$\widetilde{M_1}$ such that
$\widetilde{M_1} R\,\widetilde{M_2}$, hence 
$M_1 B\,(\widetilde{M_2}|\,M'_2)$ by definition 
of $B$.
\item
Let $(M_2|\,M'_2)[h\rangle(\widetilde{M_2}|\,M''_2)$ for 
$h\in H$, then certainly $M_2[h\rangle\widetilde{M_2}$
in $\mathcal{N}$. 
Suppose $M_1[h\rangle M_2$, then we have also 
$M_1 R\,\widetilde{M_2}$ by definition of $R$,
hence $M_1 B\,(\widetilde{M_2}|\,M''_2)$ by definition 
of $B$.
\end{itemize}
Summing up, $B$ is a weak bisimulation and $\mathcal{N}$ 
has the property BNDC.
\end{proof}

\begin{proposition}\label{jhkjb}
$\mathcal{N}$ has the property SBNDC iff
for any reachable marking $M_1$ of $\mathcal{N} = (N, M_0)$ and 
for any high-level transition $h\in H$, 
$M_1[h\rangle M_2$ entails that 
$\mathcal{L}(\mathcal{N}\setminus H,M_1)=
\mathcal{L}(\mathcal{N}\setminus H,M_2)$.
\end{proposition}

\begin{proof}
As for $\mathcal{N}\setminus H$ the labelling is the identity $\lambda(l) = l$, the thesis follows
by Proposition \ref{sim=approx}.
\end{proof}

\begin{theorem}\label{equiv}
$\mathcal{N}$ has the property BNDC {\em iff} it 
has the property SBNDC. 
\end{theorem}

\begin{proof}
Suppose that $\mathcal{N}$ has the property BNDC. 
Then , by Lemma~\ref{lem1}, $\mathcal{N}\approx\mathcal{N}
\setminus H$, hence $\mathcal{L}(\mathcal{N})=
\mathcal{L}(\mathcal{N}\setminus H)$. 
Let $M_0[s\rangle M_1$ in $\mathcal{N}$, then 
necessarily, $M_0[s'\rangle M'_1$ in  
$\mathcal{N}\setminus H$ for $s'$ defined as the
observable projection of $s$. Thus $M'_1 R\, M_1$
by definition of $R$. As $M_1[h\rangle M_2$, we have
also $M'_1 R\, M_2$. By Proposition~\ref{elabeth},
$\mathcal{L}(\mathcal{N}\setminus H,M_1)=$
$\mathcal{L}(\mathcal{N}\setminus H,M'_1)=$ 
$\mathcal{L}(\mathcal{N}\setminus H,M_2)$, hence 
$\mathcal{N}$ has the property SBNDC.

Now assume that $\mathcal{N}$ has the property
SBNDC. By Proposition~\ref{elabeth}, in order
to prove that $\mathcal{N}$ has the property BNDC,  
it suffices to show that $M_1 R\, M_2$ entails
$\mathcal{L}(\mathcal{N}\setminus H,M_1)=
\mathcal{L}(\mathcal{N}\setminus H,M_2)$ 
for all reachable markings $M_1$ and $M_2$ of 
$\mathcal{N}\setminus H$ and $\mathcal{N}$, 
respectively. Let $M_1$ and $M_2$ be two such 
markings and assume that $M_1 R\, M_2$. In view of 
Definition~\ref{R}, this relation has been  
derived from the axiom $M R\,M$ using the two 
inference rules (where we have
exchanged the $M_i$ and the $M'_i$ from 
Definition~\ref{R}):
\begin{itemize}
\item
$M'_1 R\, M'_2$ and $M'_1=M_1$ and 
$M'_2[h\rangle M_2$ entail 
$M_1 R\, M_2$
\item
$M'_1 R\, M'_2$ and $M'_1[l\rangle M_1$ and
$M'_2[l\rangle M_2$ entail $M_1 R\, M_2$
\end{itemize}
If $M_1=M_2$, then there is nothing to prove.
In the converse case, one can assume by induction
on the derivation of $M_1 R\, M_2$
that $\mathcal{L}(\mathcal{N}\setminus H,M'_1)=
\mathcal{L}(\mathcal{N}\setminus H,M'_2)$. The
desired conclusion follows then from 
Definition~\ref{sbndc} for the first rule,
and from the definition of $R$ and the injective 
labelling of nets for the second rule. 
\end{proof}

Despite the fact that SBNDC requires infinitely many equivalence checks,
one for each reachable marking enabling a high-level transition, it
(and hence also BNDC) can be decided, as will be seen
in the next section.

\subsection{Deciding SBNDC}

In this section, we reduce SBNDC to the conjunction, for all high-level
transitions $h$ and for all low-level transitions $l$, of a predicate
$P(h,l)$ meaning that the enabling or disabling of $l$ in the net after
a sequence of low transitions $s\in L^*$ gives no indication on 
whether $h$ has been fired immediately before $s$.   

\begin{definition}\label{P.def}
Given a two-level net system $\mathcal{N}$
and two transitions $h\in H$ and $l\in L$,
we say that $P(h,l)$ holds {\em iff} for any words 
$s\in L^*$ and $w\in (L\cup H)^*$,
if $M_0[w\rangle M_1$, 
$M_1[h\rangle M_2$, 
$M_1[s\rangle M_3$, and
$M_2[s\rangle M_4$, then
$M_3[l\rangle$ {\em iff}
$M_4[l\rangle$. 
\end{definition}

Figure \ref{non-P.fig} shows a situation where $P(h,l)$ is {\em not} satisfied,
because $l$ is enabled at $M_4$ but not at $M_3$.
This corresponds roughly to 
{\em causal information flow} \cite{BG09} from $h$ to $l$.
The other situation in which $P(h,l)$ is not satisfied is the symmetric one, when
$l$ is enabled at $M_3$ but disabled at $M_4$; this roughly corresponds to
{\em conflict information flow}  \cite{BG09}  from $h$ to $l$.

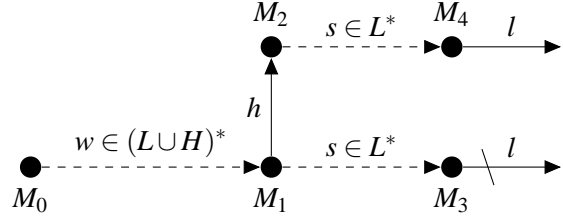
\begin{figure}
\begin{center}
\begin{tikzpicture}[scale=0.8]
\node[circle,fill=black!100,inner sep=0.1cm](M0)at(1,1)[label=below:$M_0$]{};
\node[circle,fill=black!100,inner sep=0.1cm](M1)at(5,1)[label=below:$M_1$]{};
\node[circle,fill=black!100,inner sep=0.1cm](M2)at(5,3)[label=above:$M_2$]{};
\node[circle,fill=black!100,inner sep=0.1cm](M3)at(8,1)[label=below:$M_3$]{};
\node[circle,fill=black!100,inner sep=0.1cm](M4)at(8,3)[label=above:$M_4$]{};
\node[circle,fill=black!0,inner sep=0.1cm](M3dummy)at(10,1){};
\node[circle,fill=black!0,inner sep=0.1cm](M4dummy)at(10,3){};
\draw[-triangle 45](M1.north)to[out=90,in=-90]node[auto]{$h$}(M2.south);
\draw[dashed,-triangle 45](M0.east)to[out=0,in=180]node[auto]{$w\in(L\cup H)^*$}(M1.west);
\draw[dashed,-triangle 45](M1.east)to[out=0,in=180]node[auto]{$s\in L^*$}(M3.west);
\draw[dashed,-triangle 45](M2.east)to[out=0,in=180]node[auto]{$s\in L^*$}(M4.west);
\draw[-triangle 45](M3.east)to[out=0,in=180]node[auto]{$l$}(M3dummy.west);\draw[-](8.7,0.7)--(8.5,1.3);
\draw[-triangle 45](M4.east)to[out=0,in=180]node[auto]{$l$}(M4dummy.west);
\end{tikzpicture}
\end{center}
\caption{Illustration of Property $P(h,l)$}\label{non-P.fig}
\end{figure}

\begin{proposition}
$\mathcal{N}$ has the
property SBNDC {\em iff} 
$P(h,l)$ holds 
for any high-level action $h\in H$ and for any 
low-level action $l\in L$.
\end{proposition}
\begin{proof}
This is a direct consequence of Proposition~\ref{jhkjb}.
Indeed, $M_1[h\rangle M_2$ and $P(h,l)$ for all $l$ entail
$\mathcal{L}(\mathcal{N}\setminus H,M_1)=
\mathcal{L}(\mathcal{N}\setminus H,M_2)$,
and conversely, $\mathcal{L}(\mathcal{N}\setminus H,M_1)=
\mathcal{L}(\mathcal{N}\setminus H,M_2)$ for all 
transitions $M_1[h\rangle M_2$ entail $P(h,l)$ for all $l$.
\end{proof}

We will now show that $P(h,l)$ is a decidable 
property, entailing that one can decide whether 
a given net system $\mathcal{N}$ has the property SBNDC
(because in a finite net, there are finitely many pairs $(h,l)$).
  
\begin{proposition}
$P(h,l)$ is a decidable property.
\end{proposition}
\begin{proof}
Let a net $\mathcal{N}$ with initial marking $M_0$ and two fixed transitions $h\in H$ and $l\in L$ be given.
Let $\mathcal{N}_1$ be an exact copy of $\mathcal{N}$, with place set $P_1$, except that it also contains
another `local' copy $l_1'$ of transition $l$.
Let $\mathcal{N}_2$ be another exact copy of $\mathcal{N}$, with place set $P_2$ (disjoint from $P_1$),
except that it also contains a local copy $l_2'$ of transition $l$ and a local copy $h'$ of transition $h$.
Let $\mathcal{N}'$ be defined as $\mathcal{N}_1|\mathcal{N}_2$ plus two further places $x$ and $y$
and the following extension of $F'$:
\begin{itemize}
\item[(a)]
$x$ is connected to all transitions in $H$ by a side-condition loop.
\item[(b)]
$F'(x,h')=1$, $F'(h',y)=1$, $F'(y,l_1')=1$ and $F'(y,l_2')=1$.
\end{itemize}
Finally, let $x$ be initially marked with $1$ token and $y$ with $0$ tokens.
The idea is that $\mathcal{N}'$ contains two components, one simulating the path from $M_0$ to 
$M_3$ in Figure \ref{non-P.fig}, and another one simulating the path from $M_0$ to $M_4$, if such paths exist. 

It is claimed that $P(h,l)$ holds true in $\mathcal{N}$ if and only if in the net $\mathcal{N}'$ so constructed,
it is {\em not} possible to reach a marking $M'$ such that
\begin{equation}\label{P.eq}
(M'[l_1'\rangle\wedge\neg M'[l_2'\rangle)\vee(\neg M'[l_1'\rangle\wedge M'[l_2'\rangle).
\end{equation}
To see ($\Rightarrow$), suppose that $M_0'[v\rangle M'$ where $M_0'$ is the
initial marking of $\mathcal{N}'$ defined above, and where $M'$ satisfies (\ref{P.eq}).
By (b) and because $M'$ enables either $l_1'$ or $l_2'$,
$h'$ occurs exactly once in $v$, and neither $l_1'$ nor $l_2'$ occur in $v$.
Hence $v$ can be split as $M_0'[v_1h'v_2\rangle M'$
such that $v_1$ and $v_2$ contain only transitions of $H\cup L$.
By (a), $v_2$ contains only transitions from $L$.
Because $h'$ does not change the tokens on place set $P_1$,
$v_1v_2$ is an execution sequence of $\mathcal{N}_1$,
whence $M_0[v_1v_2\rangle$ in $\mathcal{N}$.
Because $h'$ acts on $P_2$ exactly as $h$ does,
$v_1hv_2$ is an execution sequence of $\mathcal{N}_2$,
whence $M_0[v_1hv_2\rangle$ in $\mathcal{N}$.
Because $l_1'$ and $l_2'$ act on $P_1$ and $P_2$, respectively, as does $l$,
$M'_0[v_1 h' v_2 l'_1\rangle$ in $\mathcal{N'}$
         iff $M_0[v_1 v_2 l\rangle$ in $\mathcal{N}$ and
         $M'_0[v_1 h' v_2 l'_2\rangle$ in $\mathcal{N'}$
         iff $M_0[v_1 h v_2 l\rangle$ in $\mathcal{N}$.
         Because $M'$
satisfies (\ref{P.eq}), this means that $P(h,l)$ is false in $\mathcal{N}$.
More precisely, referring to Definition \ref{P.def}, putting $w=v_1$ and $s=v_2$
yields $M_0[ws\rangle M_3$ and $M_0[whs\rangle M_4$
with $\neg(M_3[l\rangle\Leftrightarrow M_4[l\rangle)$ in $\mathcal{N}$.

This argument can easily be reversed in order to prove ($\Leftarrow$).

The proof is finished because by Corollary \ref{decid.cor},
it is decidable whether or not a marking satisfying (\ref{P.eq}) is reachable in $\mathcal{N}'$.
\end{proof}

\begin{corollary}
SBNDC is decidable for finite PT-nets.
\end{corollary}

\begin{corollary}
BNDC is decidable for finite PT-nets.
\end{corollary}

Figures \ref{decid1.fig} and \ref{decid2.fig} show an example
for the construction in the preceding proof.
In Figure \ref{decid1.fig}, which depicts the net $\mathcal{N}$
with $H=\{h\}$ and $L=\{k,l\}$ on its left-hand side, we have
\[M_0[k\rangle M_3\text{ with }\neg M_3[l\rangle\text{ and }M_0[hk\rangle M_4\text{ with } M_4[l\rangle,
\]
that is, $P(h,l)$ is violated in $\mathcal{N}$.
In Figure \ref{decid2.fig}, which depicts the net $\mathcal{N}'$ resulting from the construction in the proof,
we have
\[M_0'[h'k\rangle M'\text{ with }\neg M'[l_1'\rangle\text{ and }M'[l_2'\rangle,
\]
that is, we find a reachable marking $M'$ satisfying (\ref{P.eq}).

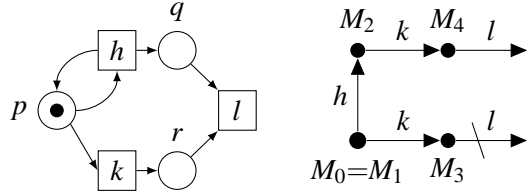
\begin{figure}[htbp]
\begin{center}
\begin{tikzpicture}[scale=0.8]
\node[circle,draw,minimum size=0.5cm](p)at(1,2)[label=left:$p$]{};\filldraw[black](1,2)circle(3pt);
\node[circle,draw,minimum size=0.5cm](q)at(3,3)[label=above:$q$]{};
\node[circle,draw,minimum size=0.5cm](r)at(3,1)[label=above:$r$]{};
\node[draw,minimum size=0.5cm](h)at(2,3){$h$};
\node[draw,minimum size=0.5cm](k)at(2,1){$k$};
\node[draw,minimum size=0.5cm](l)at(4,2){$l$};
\draw[-latex](h.east)--(q.west);
\draw[-latex](q.south east)--(l.north west);
\draw[-latex](r.north east)--(l.south west);
\draw[-latex](k.east)--(r.west);
\draw[-latex](p.south east)--(k.west);
\draw[-latex](h.west)to[out=180,in=90](p.north);
\draw[-latex](p.east)to[out=0,in=270](h.south);
\node[circle,fill=black!100,inner sep=0.08cm](M1)at(6,1.5)[label=below:$M_0{=}M_1$]{};
\node[circle,fill=black!100,inner sep=0.07cm](M2)at(6,3)[label=above:$M_2$]{};
\node[circle,fill=black!100,inner sep=0.07cm](M3)at(7.5,1.5)[label=below:$M_3$]{};
\node[circle,fill=black!100,inner sep=0.07cm](M4)at(7.5,3)[label=above:$M_4$]{};
\node[circle,fill=black!0,inner sep=0.1cm](M3dummy)at(9,1.5){};
\node[circle,fill=black!0,inner sep=0.1cm](M4dummy)at(9,3){};
\draw[-triangle 45](M1.north)to[out=90,in=-90]node[auto]{$h$}(M2.south);
\draw[-triangle 45](M1.east)to[out=0,in=180]node[auto]{$k$}(M3.west);
\draw[-triangle 45](M2.east)to[out=0,in=180]node[auto]{$k$}(M4.west);
\draw[-triangle 45](M3.east)to[out=0,in=180]node[auto]{$l$}(M3dummy.west);\draw[-](8.1,1.2)--(7.9,1.8);
\draw[-triangle 45](M4.east)to[out=0,in=180]node[auto]{$l$}(M4dummy.west);
\end{tikzpicture}
\end{center}
\caption{A system $\mathcal{N}$ violating $P(h,l)$}
\label{decid1.fig}
\end{figure}

\begin{figure}[htbp]
\begin{center}
\begin{tikzpicture}[scale=.8]
\node[circle,draw,minimum size=0.5cm](p2)at(1,4)[label=below:$p_2$]{};\filldraw[black](1,4)circle(3pt);
\node[circle,draw,minimum size=0.5cm](p1)at(3,4)[label=left:$p_1$]{};\filldraw[black](3,4)circle(3pt);
\node[circle,draw,minimum size=0.5cm](q2)at(7,7.5)[label=above:$q_2$]{};
\node[circle,draw,minimum size=0.5cm](q1)at(7,6)[label=above:$q_1$]{};
\node[circle,draw,minimum size=0.5cm](r2)at(7,0.5)[label=below:$r_2$]{};
\node[circle,draw,minimum size=0.5cm](r1)at(7,2)[label=below:$r_1$]{};
\node[circle,draw,minimum size=0.5cm](x)at(4,7.5)[label=right:$x$]{};\filldraw[black](4,7.5)circle(3pt);
\node[circle,draw,minimum size=0.5cm](y)at(4.5,4)[label=below:$y$]{};
\node[draw,minimum size=0.5cm](hd)at(2,7.5){$h'$};
\node[draw,minimum size=0.5cm](h)at(4,6){$h$};
\node[draw,minimum size=0.5cm](k)at(4,2){$k$};
\node[draw,minimum size=0.5cm](l1d)at(7.5,4){$l_1'$};
\node[draw,minimum size=0.5cm](l)at(8.5,4){$l$};
\node[draw,minimum size=0.5cm](l2d)at(9.5,4){$l_2'$};
\draw[-latex](x.west)--(hd.east);
\draw[-latex](h.north east)--(q2.south west);
\draw[-latex](h.east)--(q1.west);
\draw[-latex](q2.south east)--(l.north);
\draw[-latex](q2.east)--(l2d.north);
\draw[-latex](q1.south east)--(l.north west);
\draw[-latex](q1.south)--(l1d.north);
\draw[-latex](y.east)--(l1d.west);
\draw[-latex](r1.north)--(l1d.south);
\draw[-latex](r1.north east)--(l.south west);
\draw[-latex](r2.north east)--(l.south);
\draw[-latex](r2.east)--(l2d.south);
\draw[-latex](k.east)--(r1.west);
\draw[-latex](k.south east)--(r2.north west);
\draw[-latex](p2.south east)--(k.west);
\draw[-latex](p1.south)--(k.north);
\draw[-latex](hd.south)--(y.north west);
\draw[-latex](hd.north east)to[out=20,in=160](q2.north west);
\draw[-latex](y.south east)to[out=-30,in=210](l2d.south west);
\draw[-latex](x.south)to[out=240,in=120](h.north);
\draw[-latex](h.north)to[out=60,in=-60](x.south);
\draw[-latex](h.south west)to[out=220,in=80](p1.north);
\draw[-latex](p1.north east)to[out=40,in=-90](h.south);
\draw[-latex](h.west)to[out=180,in=70](p2.north east);
\draw[-latex](p2.north east)to[out=10,in=230](h.west);
\draw[-latex](hd.west)to[out=240,in=100](p2.north west);
\draw[-latex](p2.north)to[out=70,in=260](hd.south west);
\draw(2,0.7)node{$P_1=\{p_1,q_1,r_1\}$};
\draw(2,0)node{$P_2=\{p_2,q_2,r_2\}$};
\end{tikzpicture}
\end{center}
\caption{A system $\mathcal{N}'$ satisfying (\protect\ref{P.eq}) for some $M'$}
\label{decid2.fig}
\end{figure}
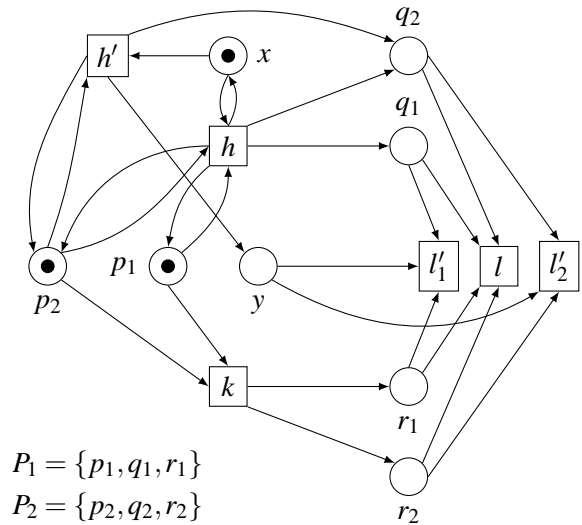

\section{Intransitive non-interference}\label{intransitive}

We enter now a less technical part of the paper, where we try 
to show how the decision results established in 
Section~\ref{noninterference} may be applied to check quality 
of control in the framework of discrete event systems. As it 
would be difficult to present applications to real systems, 
we shall consider toy examples which we hope will at least 
make the intuitions clear.

\begin{figure}[htbp]
\begin{center}
\begin{tikzpicture}[scale=.8]
\node[circle,draw,minimum size=0.5cm](s0)at(6,3.9){};
\node[circle,draw,minimum size=0.5cm](s1)at(2,1.5){};\filldraw[black](2,1.5)circle(3pt);
\node[circle,draw,minimum size=0.5cm](s2)at(6,6.5){};
\node[circle,draw,minimum size=0.5cm](s3)at(10,1.5){};
\node[circle,draw,minimum size=0.5cm](q1)at(4,6.5){};
\node[circle,draw,minimum size=0.5cm](q2)at(10,3.5){};
\node[circle,draw,minimum size=0.5cm](q3)at(4,0){};\filldraw[black](4,0)circle(3pt);
\node[draw,minimum size=0.5cm](h1)at(0,0){$h_1$};
\node[draw,minimum size=0.5cm](h2)at(6,8){$h_2$};
\node[draw,minimum size=0.5cm](h3)at(12,0){$h_3$};
\node[draw,minimum size=0.5cm](t1)at(3.5,3.5){$l_1$};
\node[draw,minimum size=0.5cm](t2)at(7.7,4.5){$l_2$};
\node[draw,minimum size=0.5cm](t3)at(6,1.5){$l_3$};
\node[draw,minimum size=0.5cm](d1)at(6,5){$d_1$};
\node[draw,minimum size=0.5cm](d2)at(7,3){$d_2$};
\node[draw,minimum size=0.5cm](d3)at(5,3){$d_3$};
\draw[-latex](d1)--(s0);
\draw[-latex](d2)--(s0);
\draw[-latex](d3)--(s0);
\draw[-latex](s1)--(d3);
\draw[-latex](s2)--(d1);
\draw[-latex](s3)--(d2);
\draw[-latex](s1)--(t1);
\draw[-latex](t1)--(s2);
\draw[-latex](s2)--(t2);
\draw[-latex](t2)--(s3);
\draw[-latex](s3)--(t3);
\draw[-latex](t3)--(s1);
\draw[-latex](h1)--(q1);
\draw[-latex](q1)--(h2);
\draw[-latex](h2)--(q2);
\draw[-latex](q2)--(h3);
\draw[-latex](h3)--(q3);
\draw[-latex](q3)--(h1);
\draw[-latex](h1)to[out=20,in=240](s1);
\draw[-latex](s1)to[out=200,in=45](h1);
\draw[-latex](h3)to[out=130,in=-20](s3);
\draw[-latex](s3)to[out=300,in=160](h3);
\draw[-latex](h2)to[out=240,in=120](s2);
\draw[-latex](s2)to[out=60,in=290](h2);
\draw[-latex](q1)to[out=340,in=130](d1);
\draw[-latex](d1)to[out=170,in=290](q1);
\draw[-latex](q2)to[out=180,in=20](d2);
\draw[-latex](d2)to[out=-10,in=210](q2);
\draw[-latex](d3)to[out=230,in=100](q3);
\draw[-latex](q3)to[out=40,in=280](d3);
\draw(4.2,4.8)node{$2$};
\draw(4.4,4.8)--(4.6,4.6);
\end{tikzpicture}
\end{center}
\caption{A three-level net system}
\label{net2}
\end{figure}
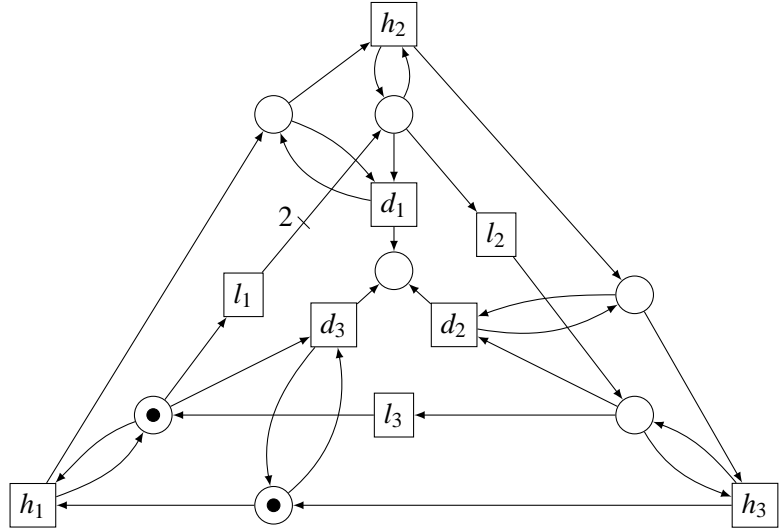

Our first example is the net system shown in Figure~\ref{net2}. 
This net is composed of two directed rings interconnected by 
bidirectional arcs plus a sink place (in the center)
fed by three transitions connected to both rings. Each arc from
a place $p$ to a transition $t$ means a flow $F(p,t)=1$. Each
arc from a transition $t$ to a place $p$ means a flow $F(t,p)=1$,
except for the arc from $l_1$ labeled with $2$, meaning that
$F(l_1,p)=2$ for the target place $p$. The internal ring formed
with the low-level transitions $l_1,l_2,l_3$ represents a flock of
prey that travel clockwise from place to place, and split each 
time they go through $l_1$. The external ring formed with the
high-level transitions $h_1,h_2,h_3$ represents an observer that 
also travels clockwise and watches the prey but moves only if 
some prey has been detected in the location currently observed. 
The three (downgrading) transitions $d_1,d_2,d_3$ represent the
actions of a predator that receives delayed notification of the 
presence of prey from the observer, and therefore anticipates 
their possible moves by one position. The objective of the
observer and predator is of course to catch prey. The 
transitions $l_1,l_2,l_3$ are scheduled by a guardian that pursues 
the opposite objective. Whenever a prey is caught, this has 
direct effect on the set of the possible schedules in 
$\{l_1,l_2,l_3\}^*$, hence there exist interferences between 
$d_1,d_2,d_3$ and $l_1,l_2,l_3$. If the set of possible schedules in 
$\{l_1,l_2,l_3\}^*$ was directly affected by the transitions in 
$h_1,h_2,h_3$, the guardian could glean information on the position 
of the observer and therefore drive the prey to safe locations. 
This is actually not the case, because the high-level 
transitions do not affect the contents of the places connected
to the low-level transitions. The fact that each $d_i$ 
transition reveals that the last transition of the observer 
was the corresponding $h_i$ makes no problem since the prey 
has already been caught. This is the essence of downgrading
transitions and intransitive non-interference in PT-nets, whose
definitions follow. 

\begin{definition}[Three-level net system]
A {\em three-level} PT-net system is a PT-net system 
$\mathcal{N}=(P,T,F,M_0)$ whose set $T$ of 
transitions is partitioned into {\em low level} 
transitions $l\in L$, {\em downgrading}
transitions $d\in D$, and {\em high level} 
transitions $h\in H$, such that $T=L\cup D\cup H$ 
and the sets $L,D$ and $H$ do not intersect.
\end{definition} 

The low-level transitions are supposed to 
be observed by the low user, while the high-level 
transitions cannot be observed and should  
hopefully be kept {\em secret}, i.e. they should not be 
revealed to the low user by the observation of the 
firing sequences in which they occur. The downgrading
transitions may be observed by the low user, but when such a
transition occurs, the requirement that all high-level 
transitions that possibly occurred before should be kept 
secret is cancelled.
This is a strong form of declassification, but we do
not know at present about the decidability of INI or BINI for more
flexible forms of declassification, where each transition $d\in D$
would declassify a corresponding subset $H_d$ of $H$ (Lemma \ref{ini-char.lem},
which is crucial to our proofs, does not apply in such a case).

\begin{definition}[INI-BINI]\label{bini} 
A three-level net system $(N,M_0)$ has the property
INI (Intransitive Non-Interference), resp. BINI 
(Bisimulation-Based Intransitive Non-Interference) 
iff the two-level net system $(N\setminus D,M)$ has
the property NDC, resp. BNDC, for $M=M_0$ and for
any marking $M$ such that $M_0[\upsilon d\rangle M$ 
in $N$ for some sequence $\upsilon\in T^*$ and for 
some downgrading transition $d\in D$.
\end{definition}

The intuition under Definition~\ref{bini} is as follows.
The {\em secret} to be covered is that some high-level 
transition $h$ has occurred {\em after the last 
downgrading transition} $d$, if any such transition was 
ever fired in $\mathcal{N}$. Whenever some downgrading 
transition $d$ is fired, the current secret is deemed
obsolete (the high-level transitions that may have 
occurred before may be revealed by the downgrading 
transition itself or by subsequent low-level 
transitions), and a new secret (namely, that some 
high-level transition may have occurred after the new 
downgrading transition) is decreed. Thus, INI (resp.
BINI) is just a clocked version of NDC (resp. BNDC),
where the ticks of the clock are the downgrading 
transitions. INI/BINI are weakenings of NDC/BNDC but 
they are still very strong security properties. We feel 
that such strong properties are really needed in the 
general context of games, including discrete event 
systems control as a particular case, where {\em any} 
piece of information leaked 
about the strategy of a player to reach its objective 
can be used by the adversary to the opposite goal. 

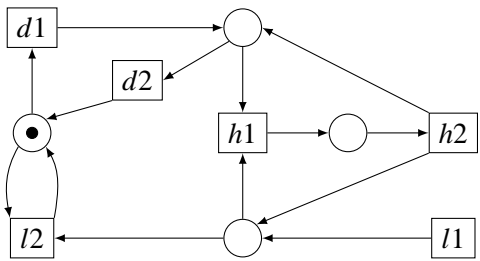
\begin{figure}[htbp]
\begin{center}
\begin{tikzpicture}[scale=0.7]
\node[draw,minimum size=0.5cm](l2)at(0,0){$l2$};
\node[draw,minimum size=0.5cm](l1)at(8,0){$l1$};
\node[draw,minimum size=0.5cm](h1)at(4,2){$h1$};
\node[draw,minimum size=0.5cm](h2)at(8,2){$h2$};
\node[draw,minimum size=0.5cm](d1)at(2,3){$d2$};
\node[draw,minimum size=0.5cm](d2)at(0,4){$d1$};
\node[circle,draw,minimum size=0.5cm](s1)at(0,2){};\filldraw[black](0,2)circle(3pt);
\node[circle,draw,minimum size=0.5cm](s2)at(4,0){};
\node[circle,draw,minimum size=0.5cm](s3)at(4,4){};
\node[circle,draw,minimum size=0.5cm](s4)at(6,2){};
\draw[-latex](s1.south west)to[out=240,in=100](l2.north west);
\draw[-latex](l2.north east)to[out=80,in=300](s1.south east);
\draw[-latex](s1.north)--(d2.south);
\draw[-latex](d2.east)--(s3.west);
\draw[-latex](s3.south west)--(d1.east);
\draw[-latex](d1.south west)--(s1.north east);
\draw[-latex](s2.west)--(l2.east);
\draw[-latex](s3.south)--(h1.north);
\draw[-latex](s2.north)--(h1.south);
\draw[-latex](h1.east)--(s4.west);
\draw[-latex](s4.east)--(h2.west);
\draw[-latex](h2.north west)--(s3.east);
\draw[-latex](h2.south west)--(s2.north east);
\draw[-latex](l1.west)--(s2.east);
\end{tikzpicture}
\end{center}
\caption{Another three-level net system}\label{net1}
\end{figure}   

In order to illustrate better non-interference in 
unbounded PT-nets, we would like to present a second 
example in which the high-level transitions do modify 
the (contents of the) input places of the low-level 
transitions. Consider the net system shown in 
Figure~\ref{net1}. The low-level transition $l1$ is 
always enabled and it represents the arrival of goods 
in a shop. The low-level transition $l2$ represents a 
sale operation and it can only be performed when the 
shop is open, which is indicated by the presence of 
one token in the leftmost place. The downgrading 
transitions $d1$ (closing the shop) and $d2$ (opening 
the shop) are operated by a guard whose friend takes
one article from the shop after closing time 
(high-level transition $h1$) and brings it back before 
opening (high-level transition $h2$). It is easily 
seen that the two high-level transitions form a 
T-invariant and that $l2$ cannot be fired between 
$h1$ and $h2$ because the shop is closed during this 
period. However, in principle, the guard's friend 
might grab the key of the shop ($h1$) immediately 
after each release (by $h2$), and this would impact 
the low view of the system since the transition $l2$ 
could then stay blocked forever (blocking may be 
perceived in weak-bisimulation based semantics). 
Our definition of BINI does not take this pathologic 
behaviour into account. Intuitively, 
Definition~\ref{bini} means that high-level transitions
are transparent to the low-level user (that is to say,
to the controlled system) unless they cause a
starvation of the downgrading transitions (that is to
say, of the controller). Therefore, the net system of 
Figure~\ref{net1} is secure w.r.t. BINI.

In the rest of the section, we show that both 
properties INI and BINI can be decided for 
unbounded PT-nets. 
$\mathcal{N}=(N,M_0)$ denotes always a 
three-level net system where $N=(P,T,F)$ and 
$T$ is partitioned into low-level transitions 
$l\in L$, high-level transitions $h\in H$, 
and downgrading transitions $d\in D$. 

\begin{lemma}\label{ini-char.lem}
$(N,M_0)$ has the property INI iff 
$(N\setminus D,M)\sim(N\setminus (H\cup D),M)$ 
for $M=M_0$ and for any marking $M$ such that
$M_0[\upsilon d\rangle M$ (in $N$) with
$\upsilon\in T^*$ and $d\in D$.
\end{lemma}

\begin{proof}
This is a direct application of Proposition~\ref{pr-ndc}.
\end{proof}

\begin{proposition}\label{ksdjfbv}
One can decide whether $(N,M_0)$ has the property
INI.
\end{proposition}

\begin{proof}
First, it can be checked whether 
$(N\setminus D,M_0)\sim(N\setminus (H\cup D),M_0)$,
because all transitions of the net system
$(N\setminus (H\cup D),M_0)$ are observable.
As a matter of fact, $\mathcal{L}((N\setminus(H\cup D),M_0))$
is always included in 
$\mathcal{L}((N\setminus D,M_0))$, and by
E.~Pelz's theorem and corollary (Theorem~\ref{Pelz} 
in the appendix), 
the reverse inclusion can be decided since  
$\mathcal{L}((N\setminus(H\cup D),M_0))$ is a 
free PT-net language.

Now fix some downgrading transition $d\in D$. Let 
$\mathcal{N}_d$ be the net system (with  
underlying net $N_d$) constructed as follows.
\begin{itemize}
\item 
$N_d$ has all places of $N$ plus two places $p_d$ 
and $p'_d$ (the complement of $p_d$). The initial 
marking 
$M_{0d}$ 
of $\mathcal{N}_d$ extends $M_0$ by 
setting one token in $p_d$ and leaving $p'_d$ 
empty. 
\item 
$N_d$ has all transitions $t$ of $N$ with flow 
relations extended by 
$F(p_d,t)=1$ and $F(t,p_d)=1$.
\item
$N_d$ has a new 
transition $d'$ with the same flow relations as 
$d$ except that $F(d',p_d)=0$ and $F(d',p'_d)=1$
(whereas $F(d,p_d)=1$ and $F(d,p'_d)=0$).
\item 
$N_d$ has a fresh copy $t'$ of each transition 
$t\in L\cup H$, with the same flow relations as 
$t$ except that $F(p'_d,t')=1$ and $F(t',p'_d)=1$ 
(whereas $F(p_d,t)=1$ and $F(t,p_d)=1$). 
\item 
all transitions of $N_d$,
including $H$ and $D$,  
are low-level transitions 
except for $H'=\{t'\,|\,t\in H\}$.
\end{itemize}
We claim
that $(N\setminus D,M)\sim(N\setminus H\cup D,M)$ 
for any $M$ such that $M_0[\upsilon d\rangle M$ 
in $N$ for the fixed $d\in D$ and for some
$\upsilon\in T^*$ iff  
$\mathcal{N}_d\sim \mathcal{N}_d\setminus H'$
(the proof of this claim, easy but a bit lengthy, 
is given in the annex, see Claim~\ref{sudhbv}).
As all transitions of $\mathcal{N}_d\setminus H'$
are observable, the language of this net system is 
a free PT-net language. It follows by E.~Pelz's
theorem and corollary (Theorem~\ref{Pelz} in the 
appendix) that one 
can decide on the inclusion relation
$\mathcal{L}(\mathcal{N}_d)\subseteq$
$\mathcal{L}(\mathcal{N}_d\setminus H')$. As there
are finitely many downgrading transitions
$d\in D$, 
by the above claim, one can decide
whether a PT-net
system has the property INI. 
\end{proof}

\begin{lemma}\label{lem-bini}
$(N,M_0)$ has the property BINI iff
for any reachable marking $M_1$ of 
$\mathcal{N}$ and for any high-level transition 
$h\in H$, $M_1[h\rangle M_2$ entails  
$\mathcal{L}(N\setminus(H\cup D),M_1)= 
\mathcal{L}(N\setminus(H\cup D),M_2)$.
\end{lemma}

\begin{proof}
By 
Proposition~\ref{jhkjb} and 
Theorem~\ref{equiv}, $(N,M_0)$ 
has the property BINI iff the following
entailment relation is satisfied
for $M=M_0$ and for any marking $M$ such 
that $M_0[\upsilon d\rangle M$ (in $N$) 
for some $\upsilon\in T^*$ and $d\in D$:\\
{\bf if} $M[w\rangle M_1$ in $N\setminus D$ 
for some $w\in(H\cup L)^*$\\ and
$M_1[h\rangle M_2$ in $N\setminus D$ 
for some $h\in H$,\\
{\bf then} $\mathcal{L}(N\setminus(H\cup D),M_1)= 
\mathcal{L}(N\setminus(H\cup D),M_2)$.\\
Grouping the case $M=M_0$ with the other 
cases, one obtains the lemma.  
\end{proof}

\begin{definition}
Given a three-level net system $\mathcal{N}$ and two
transitions $h\in H$ and $l\in L$, we say that $Q(h,l)$
holds {\em iff} for any words $\chi\in T^*$
and $s\in L^*$, if $M_0[\chi\rangle M_1$, 
$M_1[h\rangle M_2$, 
$M_1[s\rangle M_3$, and
$M_2[s\rangle M_4$, then
$M_3[l\rangle$ {\em iff}
$M_4[l\rangle$. 
\end{definition}

\begin{proposition}
One can decide whether $(N,M_0)$ has the property BINI.
\end{proposition}

\begin{proof}
By Lemma~\ref{lem-bini}, $\mathcal{N}$ has the 
property BINI iff $Q(h,l)$ holds for every high-level 
action $h$ and for every low-level action $l$. As $Q(h,l)$ 
is the same as $P(h,l)$, up to replacing $H$ with 
$H\cup D$, $Q(h,l)$ is decidable. Therefore, the 
BINI property can be decided for PT-net systems. 
\end{proof}

As nets are labeled injectively on transitions, 
$\mathcal{L}(N\setminus(H\cup D),M_1)=
\mathcal{L}(N\setminus(H\cup D),M_2)$
iff $M_1\approx M_2$ 
w.r.t. $\Sigma_o=L$.
Therefore, BINI 
coincides exactly with the property BNID specified by 
Definition 5.7 in \cite{GV09}.

\section{Conclusion and future work}\label{conc}

The examples we have discussed seem to suggest that there is a clear, structural reason why an interference
is present in a net system: either a high-level transition is causing a low-level transition (e.g., Example \ref{simple})
or a high-level transition and a low-level one are competing for the same token in a place (e.g., Example \ref{conf}).
As a matter of fact, in \cite{BG09} one of the authors showed that precisely this is the case when restricting net systems
to elementary net systems (which are essentially PT-nets where each place can contain at most one token).
More precisely, a (contact-free) elementary net system $\mathcal{N}$ is BNDC if and only if it is
never the case that a low transition consumes a token that {\em must} have been produced by a high transition nor that a high transition and 
a low-transition compete for the very same token in a place.

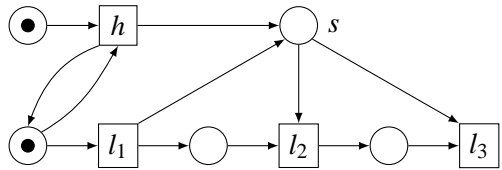
\begin{figure}[htpb]
\begin{center}
\begin{tikzpicture}[scale=0.8]
\node[circle,draw,minimum size=0.5cm](s0)at(1,3){};\filldraw[black](1,3)circle(3pt);
\node[circle,draw,minimum size=0.5cm](s)at(5.5,3)[label=right:$s$]{};
\node[circle,draw,minimum size=0.5cm](s1)at(1,1){};\filldraw[black](1,1)circle(3pt);
\node[circle,draw,minimum size=0.5cm](s2)at(4,1){};
\node[circle,draw,minimum size=0.5cm](s3)at(7,1){};
\node[draw,minimum size=0.5cm](h)at(2.5,3){$h$};
\node[draw,minimum size=0.5cm](l1)at(2.5,1){$l_1$};
\node[draw,minimum size=0.5cm](l2)at(5.5,1){$l_2$};
\node[draw,minimum size=0.5cm](l3)at(8.5,1){$l_3$};
\draw[-latex](s0.east)--(h.west);
\draw[-latex](h.east)--(s.west);
\draw[-latex](s1.east)--(l1.west);
\draw[-latex](l1.east)--(s2.west);
\draw[-latex](s2.east)--(l2.west);
\draw[-latex](l2.east)--(s3.west);
\draw[-latex](s3.east)--(l3.west);
\draw[-latex](l1.north east)--(s.south west);
\draw[-latex](s.south)--(l2.north);
\draw[-latex](s.south east)--(l3.north west);
\draw[-latex](h.south west)to[out=200,in=70](s1.north);
\draw[-latex](s1.north east)to[out=30,in=240](h.south);
\end{tikzpicture}
\end{center}
\caption{A non BNDC net}\label{counterexample}
\end{figure}

Unfortunately, generalizing this characterization in the setting of general PT-nets seems problematic.
Consider the net system $\mathcal{N}$ shown in Figure \ref{counterexample}.
Let $M_0$ be the initial marking indicated in the figure. Set $M_0  [h\rangle M_1$ and set also $M_0 [l_1 l_2 \rangle M_2$ and
$M_1 [l_1 l_2 \rangle M_3$. Clearly, transition $l_3$ is enabled at $M_2$ but disabled at $M_3$, hence $\mathcal{N}$ is not BNDC. However, in the firing sequence $M_0 [h l_1 l_2 l_3 \rangle$, the token consumed from place $s$ by the low-level transition $l_3$ {\em may} have been produced by the high-level transition $h$ but it {\em may} also have been produced alternatively by the low-level transition $l_1$.

As regards continuations of this work, it would be useful to look at flexible versions of downgrading, where each downgrading action bears upon a specific subset of high-level actions. A wider perspective would be to investigate non-interference in the framework of games of partial information, see e.g. \cite{RESDSW10} for a survey on Games for Security.

\section*{Acknowledgment}

The authors would like to thank the reviewers for their comments.

\section{Annex}

\begin{definition}[PT-nets]
A {\em PT-net} is a bi-partite graph $N=(P,T,F)$, 
where $P$ and $T$ are {\em finite} disjoint sets of vertices, 
called {\em places} and {\em transitions}, 
respectively, and $F:(P\times T)\cup(T\times P)
\rightarrow\bbbn$ is a set of directed edges with
non-negative integer weights. A {\em marking} of 
$N$ is a map $M:P\rightarrow\bbbn$. A transition 
$t\in T$ is {\em enabled at} a marking $M$ 
(notation:$M[t\rangle$) if $M(p)\geq F(p,t)$ for 
all places $p\in P$. It $t$ is enabled at $M$, 
then it can {\em be fired}, leading to the new
marking $M'$ (notation: $M[t\rangle M'$) defined 
by $M'(p)=M(p)+F(t,p){-}F(p,t)$ for all $p\in P$.
These definitions are extended inductively to 
transition sequences $s\in T^*$: 
for the empty sequence $\varepsilon$,
$M[\varepsilon\rangle$ and $M[\varepsilon\rangle M$
are always true; for a non-empty sequence 
$s t$ with $t\in T$, $M[s t\rangle$ 
(or $M[s t\rangle M'$) iff 
$M[s\rangle M''$ and $M''[t\rangle$ 
(or $M''[t\rangle M'$, respectively) for some $M''$.
A marking $M'$ is {\em reachable} from a marking 
$M$ if $M[s\rangle M'$ for some $s\in T^*$. 
The set of markings reachable from $M$ is denoted by 
$[M\rangle$.
\end{definition}

\begin{theorem}[Mayr \cite{Mayr}]\label{Mayr.thm}
Given a PT-net $N$ and two markings $M$ and $M'$,
one can decide whether $M'$ is reachable from $M$. 
\end{theorem}

\begin{definition}[Free language of a net system]
The {\em free language} of a Petri net system
$\mathcal{N}$ is the language of the LTS
$RG(\mathcal{N})$, where all transitions are considered 
observable, i.e., $\Sigma_o=T$.
In this case, we write $\mathcal{L}(\mathcal{N})$ to denote the free 
language.
\end{definition}

\begin{theorem}[Pelz \cite{Pelz}]\label{Pelz} 
The complement in $\Sigma_o^*$ of the free language of 
a net system may be generated by a labeled net 
$(\mathcal{N},\lambda)$ with a finite set of final 
partial markings, characterized by a formula $\mathcal{F}$ 
built from the logical connectives $\wedge$ and $\vee$ 
and atomic formulas $M(p)=i$ 
(with $p\in P$ and $i\in\bbbn$). In other words, a sequence 
$s\in\Sigma_o^*$ belongs to this complement 
if and only if $s=\lambda(t_1 t_2\ldots t_n)$ for some 
sequence of transitions $M_0[t_1t_2\ldots t_n\rangle M$ 
of $\mathcal{N}$ such that $M$ satisfies $\mathcal{F}$. 
\end{theorem}

\begin{corollary}[Pelz]
The problem whether the language of a labeled
net system $\mathcal{N}_1$ is included in the free 
language of a net system $\mathcal{N}_2$ is decidable. 
\end{corollary}

\begin{proof}
The language of $\mathcal{N}_1$ is included in the 
free language of $\mathcal{N}_2$ if and only if
no marking satisfying $\mathcal F$ can be reached 
in $\mathcal{N}_1\,|\,\mathcal{N}'_2$ where 
$\mathcal{N}'_2$ is the complementary net of 
$\mathcal{N}_2$ and $\mathcal F$ is the logical
formula defining the final partial markings of 
$\mathcal{N}'_2$.
The latter reachability property can be decided in view 
of the Proposition~\ref{semi-linear-decidable.prop} 
recalled below in this appendix.
\end{proof}

In order to make the statement of 
Proposition~\ref{semi-linear-decidable.prop} 
understandable, let us recall first the basics of 
semi-linear sets and their decidable properties. 
Given a number $n\in\bbbn$, we consider the commutative monoid $(\bbbn^n,+)$
where $+$ denotes the componentwise addition of $n$-vectors and the null $n$-vector is the neutral element.
Typically, $n$ is the number of places of a Petri net and then $\bbbn^n$ is the realm
of all possible markings of this net (markings are seen as vectors in which each entry defines the number of tokens in the corresponding place for some fixed enumeration of the places of the net).

A subset $E\subseteq\bbbn^n$ is called {\em linear} if it is of the form
\[E=\{a+k_1{\cdot}b_1+\ldots+k_m{\cdot}b_m\mid k_1,\ldots,k_m\in\bbbn\}
\]
for some specific vectors $a\in\bbbn^n$ and $b_1,\ldots,b_m\in\bbbn^n$.
For example, let an unmarked net with $n$ places and a transition $t$ be given.
Then the set of markings enabling $t$ is linear, since any such marking $M$ can
be expressed as the following sum:
\[M\;\;=\;\;M_t+k_1{\cdot}b_1+\ldots+k_n{\cdot}b_n
\]
where $M_t$ is the (unique!) minimal marking enabling $t$ and the $b_1,\ldots,b_n$
are the unit vectors corresponding to the places of the net.
The natural numbers $k_1,\ldots,k_n$ simply describe excess tokens
which may be present in $M$ but are not needed for enabling $t$.

A subset $E\subseteq\bbbn^n$ is called {\em semi-linear}
if it is a finite union of linear sets. For example, if $t_1$ and $t_2$ are two transitions,
then the set of markings enabling $t_1$ or $t_2$ (or both) is semi-linear, since it is
the union of the set of markings enabling $t_1$ and the set of markings enabling $t_2$.

\begin{theorem}[Ginsburg and Spanier \cite{GS64}]\label{GS.thm}
The semi-linear subsets of $\bbbn^n$ form an effective boolean algebra.
\end{theorem}

Thus, 
if $E$, $E_1$ and $E_2$ are semi-linear subsets of $\bbbn^n$,
then so are $\bbbn^n{\setminus}E$, $E_1\cap E_2$ and $E_1\cup E_2$.
The effectiveness part of Ginsburg and Spanier's theorem concerns the
possible description of semi-linear sets as linear expressions, and it states that
the expressions of a composed set (such as $E_1\cap E_2$) can be computed effectively from
the linear expressions of the constituent set(s) (such as $E_1$ and $E_2$).


\begin{proposition}\label{semi-linear-decidable.prop}
Given a PT-net system $\mathcal{N}=(P,T,F,M_0)$ and a 
semi-linear subset of markings $E\subseteq\mathbb{N}^n$, where $n=|P|$,
one can decide whether (some marking in) $E$ can be reached from $M_0$.  
\end{proposition}

The above proposition follows from Lemma 4.3 in \cite{Hack2} where 
the semi-linear reachability problem is reduced to the reachability 
problem, and from  Theorem \ref{Mayr.thm}.


In this paper, we use Proposition \ref{semi-linear-decidable.prop}
and Theorem \ref{GS.thm} in the special form as follows.
 

\begin{corollary}[]\label{decid.cor}
Let $\mathcal{N}$ be a PT-net system with initial marking $M_0$ and let $t_1$ and $t_2$ be two transitions.
The question whether there is some marking $M\in[M_0\rangle$ with
\begin{equation}\label{decid-cor.eq}
(\;M[t_1\rangle\wedge\neg\ M[t_2\rangle\;)\;\vee\;(\;\neg\ M[t_1\rangle\wedge M[t_2\rangle\;) 
\end{equation}
is decidable.
\end{corollary}

\begin{proof}
The set of all markings $M$ satisfying (\ref{decid-cor.eq}) is semi-linear.
This follows from Theorem \ref{GS.thm}, together with the fact that
the set of markings enabling a single transition is linear.
The claim now follows directly from Proposition \ref{semi-linear-decidable.prop}.
\end{proof}

We finally give a detailed proof of the claim made in the 
proof of Proposition~\ref{ksdjfbv}. 

\begin{claim}\label{sudhbv}
With the notations used in the proof of 
Proposition~\ref{ksdjfbv}
$(N\setminus D,M)\sim(N\setminus H\cup D,M)$ 
for any $M$ such that $M_0[\upsilon d\rangle M$ 
in $N$ for some
$\upsilon\in T^*$ {\bf iff}  
$\mathcal{N}_d\sim \mathcal{N}_d\setminus H'$.
\end{claim}

\begin{proof}
We need examining closely the relationship between the 
firing sequences of $N$ and $N_d$.
Let $M_0[\upsilon d\rangle M$ be a firing sequence of $N$
and let $M[t_1\ldots t_n\rangle$ be a firing sequence of
$N\setminus D$. Then $M_{0d}[\upsilon d\rangle M_d$
in $\mathcal{N}_d$ where $M_d(p_d)=1$, $M_d(p'_d)=0$, and
$M_d(p)=M(p)$ for every place $p$ of $N$. Clearly, 
$M_d[t_1\ldots t_n\rangle$ is a firing sequence of 
$N_d\setminus D$. In a similar way, 
$M_{0d}[\upsilon d'\rangle M'_d$ in $\mathcal{N}_d$ where 
$M'_d(p_d)=0$, $M'_d(p'_d)=1$, and $M'_d(p)=M(p)$ for 
every place $p$ of $N$. Also clearly, 
$M'_d[t'_1\ldots t'_n\rangle$ 
is a firing sequence of $N_d\setminus D$.      
Conversely, consider now a firing sequence $M_{0d}[u\rangle$
in $\mathcal{N}_d$. If $d'$ does not occur in $u$, then
$M_0[u\rangle$ in $N$. If $u=\upsilon d'w$, then 
necessarily, $M_0[\upsilon d\rangle M$ for some $M$ in $N$,
and $w=t'_1\ldots t'_n$ for some sequence 
$t_1\ldots t_n\in (L\cup H)^*$ such that
$M[t_1\ldots t_n\rangle$ in $N$ and hence also in 
$N\setminus d$. 

Suppose that $(N\setminus D,M)\sim(N\setminus H\cup D,M)$
for any $M$ such that $M_0[\upsilon d\rangle M$ in $N$ for 
the fixed $d\in D$ and for some $\upsilon\in T^*$.
By construction, any sequence of transitions of  
$\mathcal{N}_d$ not including $d'$ is also a sequence of
transitions of $\mathcal{N}_d\setminus H'$. Now any 
sequence of transitions of $\mathcal{N}_d$ including $d'$
is of the form $M_{0d}[\upsilon d't'_1\ldots t'_n\rangle$, where
no transition from $H'$ occurs in $\upsilon$ and
$t'_1\ldots t'_n$ is the primed version of some sequence
$t_1\ldots t_n\in (L\cup H)^*$. Then, 
$M_0[\upsilon d\rangle M$ and $M[t_1\ldots t_n\rangle$ for
some $M$ in $N$. For all $t_j$ let 
$\lambda(t_j)=\varepsilon$ if $t_j\in H$ and 
$\lambda(t_j)=t_j$ otherwise. As    
$(N\setminus D,M)\sim(N\setminus H\cup D,M)$, one has also
$M[\lambda(t_1)\ldots\lambda(t_n)\rangle$. Therefore, if
we let $\lambda'(t'_j)=\varepsilon$ if $t'_j\in H'$ and 
$\lambda'(t'_j)=t'_j$ otherwise, then 
$M'_d[\lambda'(t'_1)\ldots \lambda'(t'_n)\rangle$ in
$N_d\setminus D$ where $M'_d$ is the marking of 
$\mathcal{N}_d$ defined with $M'_d(p_d)=0$, $M'_d(p'_d)=1$, 
and $M'_d(p)=M(p)$ for every place $p$ of $N$. As no 
transition from $H'$ occurs in 
$\upsilon d'\lambda'(t'_1)\ldots\lambda'(t'_n)$, this 
sequence is a firing sequence of $\mathcal{N}_d\setminus H'$. 
Thus, $\mathcal{N}_d\sim \mathcal{N}_d\setminus H'$.

In order to establish the converse implication, suppose now
that $\mathcal{N}_d\sim\mathcal{N}_d\setminus H'$.
Consider any two firing sequences $M_0[\upsilon d\rangle M$ 
and $M[t_1\ldots t_n\rangle$ of $N$ with 
$t_1\ldots t_n\in (L\cup H)^*$. By construction of 
$\mathcal{N}_d$, $M_{0d}[\upsilon d't'_1\ldots t'_n\rangle$. 
As no transition from $H'$ occurs in $\upsilon$, by the above 
assumption, 
$M_{0d}[\upsilon d'\lambda'(t'_1)\ldots\lambda'(t'_n)\rangle$ 
in $\mathcal{N}_d\setminus H'$ where
$\lambda'(t'_j)=\varepsilon$ if $t'_j\in H'$ and 
$\lambda'(t'_j)=t'_j$ otherwise.
Thus,
if we set $\lambda(t_j)=\varepsilon$ if $t_j\in H$ and 
$\lambda(t_j)=t_j$ otherwise, then 
$M_{0d}[\upsilon d\lambda(t_1)\ldots\lambda(t_n)\rangle$
by construction of $\mathcal{N}_d$.     . 
As a consequence, $M[\lambda(t_1)\ldots\lambda(t_n)\rangle$ 
in $N$ and hence also in
$N\setminus H\cup d$, concluding the proof of the claim.
\end{proof}
    
\end{document}